\documentclass[12pt, a4paper]{article}

\usepackage[latin1]{inputenc}
\usepackage{vmargin}
\usepackage{amsfonts}
\usepackage{amssymb}
\usepackage{graphicx}
\usepackage{amsmath}
\usepackage{amsthm}
\usepackage{setspace}
\usepackage{multirow}
\usepackage{hyperref}
\usepackage{cases}
\usepackage[round]{natbib}

\newtheorem{theorem}{Theorem}[section]
\newtheorem{lemma}{Lemma}[section]
\newtheorem{proposition}{Proposition}[section]
\newtheorem{remark}{Remark}[section]

\newtheorem{assumption}{Assumption}[section]

\newtheorem{example}{Example}

\let\Item\item

\renewcommand\enddescription{\endlist\global\let\item\Item}

\newcommand{\E}{\mathbb{E}}

\def\one{\mathbf{1}}

\usepackage[bottom]{footmisc}
\setmarginsrb{2cm}{3cm}{2cm}{2cm}{0cm}{0cm}{0cm}{0.5cm}
\title{Ergodicity and diffusivity of Markovian\\
 order book models: a general framework}

\author{Weibing Huang$^{1,2}$ and Mathieu Rosenbaum$^{1}$\\$~~$\\
$^{1}$ LPMA, University Pierre et Marie Curie (Paris 6)\\
$^{2}$ Kepler-Cheuvreux}

\date{\today}

\begin{document}

\maketitle

\begin{abstract}
\noindent We present a general Markovian framework for order book modeling. Through our approach, we aim at providing a tool enabling to get a better understanding of the price formation process and of the link between microscopic and macroscopic features of financial assets. To do so, we propose a new method of order book representation, and decompose the problem of order book modeling into two sub-problems: dynamics of a continuous-time double auction system with a fixed reference price; interactions between the double auction system and the reference price movements. State dependency is included in our framework by allowing the order flow intensities to depend on the order book state. Furthermore, contrary to most existing models, the impact of the order book updates on the reference price dynamics is not assumed to be instantaneous. We first prove that under general assumptions, our system is ergodic. Then we deduce the convergence towards a Brownian motion of the rescaled price process. 
\end{abstract}

\section{Introduction}

Nowadays, most financial exchanges use a limit order book (LOB) mechanism. In these order-driven markets, market participants send their buy and sell orders to a continuous-time double auction system, where orders are matched according to their price and time priority. Understanding the LOB dynamics is one of the fundamental issues in the field of market microstructure and leads to many interesting applications in optimal execution, design of electronic trading algorithms, minimization of market impact costs, short-term predictions and regulation. In the recent years, many works have been devoted to the description of order book dynamics. Order book models can be essentially divided into two types: economic models, where one focuses on the behaviors of individual agents and their optimal decisions, see for example \citet*{parlour1998price}, \citet*{foucault1999order} and \citet*{rocsu2009dynamic}; statistical models, where the order flows are seen as random processes, see \citet*{smith2003statistical}, \citet*{cont2010stochastic}, \citet*{abergel2011mathematical}, \citet*{Cont2010}, \citet*{Lakner2013}, \citet*{Lachapelle2014}, \citet*{Bayer2015} and \citet*{abergel2015}. With the notable exception of  \citet*{abergel2015}, where the authors consider the case of Hawkes-type dynamics, these models usually assume Poisson flows for the order arrival processes. Such assumption is mainly made for technical reasons, since it is well-known that it is not consistent with market data. In \citet*{huang2013simulating}, the authors propose to replace the Poisson assumption by a state-dependent approach where the intensities of the flows depend on the state of the LOB. This model, called ``Queue-reactive'' model, provides new insights for the order book dynamics, such as market participants behaviors conditional on different states of the order book, the LOB's asymptotic form and the bid-ask spread distribution. It is also a very relevant tool for practitioners in the perspective of transaction cost analysis of complex trading algorithms. In this paper, we aim at extending the Queue-reactive model to a more general framework, in which most of the existing statistical models can be included (up to minor modifications). Our goal is to give some theoretical results on the system's ergodicity as well as the asymptotic scaling limit of the price process. \\

In the LOB, price levels are discretized by a minimum price change unity called the tick value (denoted by $\alpha$). Market participants can place their buy/sell orders at any level which is a multiple of the tick value and these orders will either stay in the LOB (a buy order with price lower than the current best ask price, or a sell order with price higher than the current best bid price, this type of orders being called ``limit order''), or be matched with the existing orders in the LOB (this type of orders being called ``market order''). The LOB, as its name suggests, is composed of all unmatched limit orders and can be seen as a (rough) approximation of the current microstructural supply and demand on the different price levels. \\

Current statistical models differ in their way of representing the LOB. In \citet*{cont2010stochastic}, the price grid is supposed to be finite $(n_{min}\alpha,...,n_{max}\alpha)$, and the LOB is represented by a $n_{max}-n_{min}+1$ dimensional vector that records the buying/selling quantities at each of these price levels. In such representation, the different limits are indexed by their absolute price level. In practice, to cover the intra-day price range $[p_{min}, p_{max}]$, the dimension of the state space have to be at least $\frac{p_{max}-p_{min}}{\alpha}+1$\footnote{Of course if we choose $n_{min} = 1$ and $n_{max} = \infty$, we have a complete description of all the available buying/selling offers in the whole price grid $\alpha \mathbb{N}$. }, which is typically a very large number. Another way of representing the order book state is to use the relative indexing method. Following ideas from the Zero-intelligence model of \citet*{smith2003statistical}, \citet*{abergel2011mathematical} propose to use the best bid and the best ask prices as two reference prices to index the limits. In that case, the LOB is made of the following elements: the two reference prices $p_{bestbid}$ and $p_{bestask}$, and the limits around them which are two $K$ dimensional vectors $a$ (for the ask side) and $b$ (for the bid side). The vector $a = [a_1,...,a_K]$ records the limit sizes at the price levels $[p_{bestbid} + \alpha, ..., p_{bestbid} + K\alpha]$, while the vector $b = [b_1,...,b_K]$ records the limit sizes at the price levels $[p_{bestask} - \alpha, ..., p_{bestask} - K\alpha]$\footnote{Note that in this representation, there can be overlaps between the vector $a$ and $b$ when the bid-ask spread is larger than one tick.}. In practice, observing a market depth of five ticks is usually considered enough for most trading purposes. Consequently, for a typical stock with spread size of order five ticks, the value of $K$ should be generally of order 10 so that essential information from the LOB is captured. Thus, the use of these two reference prices reduces significantly the dimensionality of the state space. Note that in this representation, the index of a limit at a given price level is no longer constant. Therefore, appropriate boundary conditions must be defined to deal with price changes. \\ 

In this paper, we propose an original representation of the order book, using only one reference price which is not necessarily directly observable from the order book state. We view this reference price as sort of market consensus about the underlying ``efficient" price used by market participants when making their trading decisions. We keep $K$ limits on each side of the reference price and the LOB is fully described by a $2K+1$ dimensional vector, which is then modeled by a continuous-time Markov jump process. The use of this unique reference price gives us a lot of flexibility when modeling the order book. Since the reference price is no longer directly determined by the order book state, we can differentiate two types of jumps in the Markov process: pure order book state jumps (jumps for which the reference price stays constant) and common jumps (jumps for which a reference price change is involved). For large tick assets\footnote{A large tick asset is defined as an asset whose bid-ask spread is almost always equal to one tick.}, such decomposition is proved to be very relevant when studying the conditional dependences between the dynamics of the LOB and its state, see \citet*{huang2013simulating}. Moreover, in this framework, we are able to easily incorporate exogenous price movements into the order book dynamics. This can be simply done adding a reference price jump component which is independent of the order book state. \\

At the high frequency scale, the LOB state is one of the two public information that are accessible to traders and their automates (the other being the history of the order flows). Thus it plays a very important role in their trading decisions. In our framework , the LOB is assumed to be a continuous-time Markov jump process, and the influence of the LOB state on the incoming flows is modeled through a state-dependent infinitesimal generator matrix for the jump process. Indeed, in practice, traders essentially rely on information deduced from the current LOB state when deciding to send an order at a specific price level. Various simplifying assumptions on the information set used by traders can be considered in our framework in order to facilitate the empirical studies. The index of a limit, for example, is probably one of the most important elements in their decision process, as it gives the distance between the target price and the reference price. Influence of other variables, such as the target limit's size, its relative distance to the current best offer queues and the size of its opposite queue is studied in \citet*{huang2013simulating}, and are shown to also have non-negligible effects on the dynamics of the order flows. \\

Under appropriate assumptions, the Queue-reactive model can be easily estimated using empirical data. It provides many interesting new insights on the origin of some micro-structural properties, such as the stylized empirical distribution of the LOB state. It has been shown in \citet*{Bouchaud2002statistical} that there exist some regularities in the order book's empirical form, that is the average value of the LOB state (a $2K$ dimensional vector in our model). From a theoretical point of view, these regularities are closely linked with the notion of ergodicity of the LOB system (the exact definition of ergodicity will be given in Section \ref{vue}): Ergodicity ensures the convergence of the LOB state distribution towards an invariant probability measure. Thus the stylized form observed on market data might be explained by a law of large number type phenomenon for this invariant distribution. This hypothesis is supported by empirical studies in \citet*{huang2013simulating}, in which the authors compare the theoretical asymptotic distributions in our model with empirical estimations, and show that they are very close. In \citet*{huang2013simulating}, some assumptions are made to ensure the LOB system's ergodicity. In this paper, we want to generalize them in an extended framework where the volume of the orders is no longer constant and the influence of the order book state on the dynamics of the reference price may not be instantaneous. \\

Another important element in order book modeling is the asymptotic behavior of the price. Such analysis is very relevant as it provides useful insights on the price formation process, and links the dynamics at the microscopic level with macroscopic features of the asset, such as its volatility. We prove that in our framework, the rescaled reference price process converges to a Brownian motion. An expression for the macroscopic volatility in terms of the flow rates is derived using a functional central limit theorem together with the strong mixing property of the price increments, in the spirit of \cite{abergel2011mathematical}. \\

The paper is organized as follows. In Section 2, we set up the general Markovian framework. Section 3 discusses the ergodicity properties of the model. The diffusive limit of the rescaled price process is stated in Section 4. Finally, in Section 5, we give some specific models which can be seen as particular cases of our framework. The technical proofs are gathered in an appendix.

\section{A general Markovian framework}

\subsection{Representation of the order book}

In our framework, the order book is made of two elements: the center position: a certain reference price (note that here the center position is not necessarily the mid price) and the shape of the book (the queue sizes around the reference price), see Figure \ref{fig1} for an example. The center position, denoted by $p_{ref}\in \{n\alpha+0.5\alpha, n \in \mathbb{Z}\}$\footnote{For the generality of the framework, we allow for negative prices, however, in practice, the model should of course be used over a reasonably small time interval so that prices remain positive.}, can be seen as the current consensus price level and is used to index the limits. We write $Q_{\pm i}$ for the queue at the price level $p_{ref} \pm (i-0.5)\alpha$ and denote its size by $q_{\pm i}$ and its price by $p_i$. Then the total order book's shape at time $t$ is an infinite vector $q(t) = [...,q_{-k}(t),...q_{-1}(t),q_1(t),...,q_k(t),...]$, with $q_i \in \mathbb{Z}$, and $|q_i|$ is the number of orders at the limit $Q_i$ ($q_i < 0$ if these orders are bid orders and $q_i > 0$ if they are ask orders. Note that in such a representation, one can have $q_i\geq 0$ or $q_i < 0$ for all $i$). The LOB information at time $t$ is therefore fully represented by $X(t) := (q(t), p_{ref}(t))$, $t\geq 0$. To restrict the value of $X(t)$, we consider only $K$ limits on each side, and thus have now $q(t) = [q_{-K}(t),...q_{-1}(t),q_1(t),...,q_K(t)]$ and $X(t) \in \mathbb{Z}^{2K} \times \alpha(0.5+\mathbb{Z})$. \\

\begin{figure}
\centering
\includegraphics[scale=0.6]{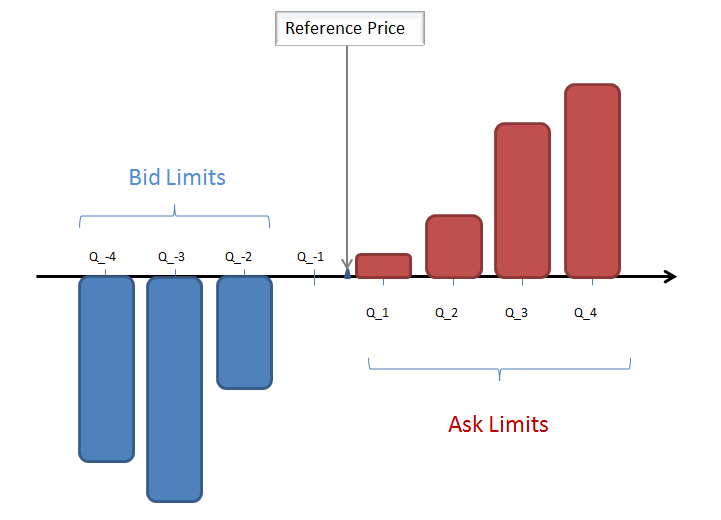}
\caption{Limit order book}\label{fig1}
\end{figure}

The state space of $q(t)$ is actually smaller than $\mathbb{Z}^{2K}$. More specifically, let us define
\begin{eqnarray}
i_{bestbid}(q) &=& \mbox{max}(-K-1, \mbox{sup}\{i | q_i < 0\}) \nonumber \\
i_{bestask}(q) &=& \mbox{min}(K+1, \mbox{inf}\{i | q_i > 0\}). \nonumber
\end{eqnarray}

The state space $\Omega$ of $q(t)$ is defined as the set of all $q \in \mathbb{Z}^{2K}$, such that: for all $i \in \{-K,...,K\}$, $q_i \leq 0$ if $i \leq i_{bestbid}(q)$ and $q_i \geq 0$ if $i \geq i_{bestask}(q)$.\\

\subsection{Dynamics of the order book}

In our general framework, we model the LOB vector $X(t)$ as a continuous-time Markov jump process, whose infinitesimal generator matrix $\mathcal{Q}$ will be given in Equation \eqref{eqa4}. We differentiate two types of jumps in the order book dynamics: pure order book state jumps (for which $p_{ref}$ stays constant) and common jumps (for which a change in the value of $p_{ref}$ is involved). 

\subsubsection*{Pure order book state jump}

There are three\footnote{Four if we also consider the modification orders. We view modification orders as a combination of a cancellation and an insertion order that arrive in a very short time interval.} types of orders that interact directly with the order book and trigger pure order book state jumps: 
\begin{itemize}
\item Limit orders: insertion of a new order in the order book (a buy order at a lower price than the best ask price, or a sell order at a higher price than the best bid price).
\item Cancellation orders: cancellation of an already existing order in the order book.
\item Market orders: consumption of available liquidity (a buy or sell order at the best available price).
\end{itemize}

In the seminal work of \citet*{smith2003statistical}, the arrival times of the above three types of orders at different price levels are assumed to be mutually independent and exponentially distributed. Furthermore, each order has unit size. In our approach, the size of the jumps, which represents the amount of volume inserted to/removed from the LOB for a given event, is random. Moreover, the arrival rate of a given jump is assumed to be function of the index of the target price, the current LOB state vector $q(t)$, the direction of the jump and its size. That is, for any $q,q' \in \Omega$ ($q \neq q'$), $p \in \alpha(0.5+\mathbb{Z})$, $n \in \mathbb{N}^{+}$ and any $e_i = (a_{-K},...,a_{i},...,a_{K})$ ($a_j = 0$ for $j \neq i$ and $a_{i} = 1$), we have in cases where $q+ne_i \in \Omega$ and $q-ne_i \in \Omega$: 
\begin{eqnarray}
\mathcal{Q}_{(q,p),(q+ne_i,p)} &=& f_i(q,n) \nonumber\\
\mathcal{Q}_{(q,p),(q-ne_i,p)} &=& g_i(q,n)  \nonumber\\
\tilde{\mathcal{Q}}_{(q,p),(q',p)} &=& 0 \mbox{, otherwise,} \label{eqa1}
\end{eqnarray}
where the $f_i$ and $g_i$ are $2K$ functions: $\Omega \times \mathbb{N}^{+} \to \mathbb{R}^{+}$. 

\noindent Note that in \eqref{eqa1}, $f_i(q,n)$ and $g_i(q,n)$ have different meanings for different $i$ and $q$. For example, when $i \geq i_{bestask}$, $f_i(q,n)$ represents the arrival rate of sell limit orders of size $n$, and $g_i(q,n)$ the sum of the rate of cancellations of size $n$ and the arrival rate of market buy orders of size $n$. When $i \leq i_{bestbid}$, the role of $f_i(q,n)$ and $g_i(q,n)$ are switched. Note also that $q\pm ne_i$ is not always in the state space $\Omega$ even when $q \in \Omega$. Thus some values of the functions $f_i$ and $g_i$ are not needed in Equation \eqref{eqa1} and assumed to be equal to zero. Furthermore, so that there is no absorbing state, we assume
$$\sum_i\sum_n\big(f_i(q,n)+g_i(q,n)\big)>0.$$

\subsubsection*{Common jumps}

The reference price $p_{ref}$ can be viewed as a consensus value on the ``efficient" price and takes discretized values in $\alpha (0.5+\mathbb{Z})$. In practice, this reference price is built based on two sets of information: the current state of the LOB and the historical order flows. We have in mind that $p_{ref}$ moves in a Markovian manner, so its dynamics depends on the present information only, that is the current state of the LOB. In our framework, we restrict the price jump size to one tick at each time. We use two functions $u,d: \Omega \to \mathbb{R}^{+}$ to describe respectively the rate of positive and negative jumps: 

\begin{eqnarray}
\sum_{q' \in \Omega} \mathcal{Q}_{(q,p),(q',p+\alpha)} &=& u(q) \nonumber\\
\sum_{q' \in \Omega} \mathcal{Q}_{(q,p),(q',p-\alpha)} &=& d(q) \nonumber\\
\sum_{q' \in \Omega} \mathcal{Q}_{(q,p),(q',p \pm n\alpha)} &=& 0, \mbox{ for } n \geq 2. \label{eqa2}
\end{eqnarray}


To understand Equation \eqref{eqa2}, let us first consider the following simple example where the LOB state information is summarized by the difference between the current value of $p_{ref}$ and the mid price $p_{mid}$: 

\begin{example}\label{ex1}
\begin{eqnarray}
i_{mid} &=& ({i_{bestbid}}+{i_{bestask}})/2, \nonumber \\
u(q) &=& \theta_0+\theta_1 \max(0,\alpha( i_{mid} - 0.5)) \nonumber \\
d(q) &=& \theta_0+\theta_1 \max(0,-\alpha( i_{mid} + 0.5 )) \nonumber,
\end{eqnarray}
with $\theta_0\geq 0$ and $\theta_1$ a positive constant representing the intensity of the adjustment of $p_{ref}$ towards $p_{mid}$.
\end{example}

In the above example, we assume that the reference price jump rate depends on the deviation of the current value of $p_{ref}$ from $p_{mid}$. Indeed, $p_{mid}$ is often considered an approximation of the LOB center implied by its current state. We may also include other LOB information such as the available quantities at $Q_{i_{bestbid}}$ and $Q_{i_{bestask}}$ when defining $u(q)$ and $d(q)$. Such additional variables increase the complexity of our model but make it more realistic. 


\begin{remark}
In the already mentioned Queue-reactive model, changes of $p_{ref}$ are triggered by the order book events that modify $p_{mid}$, while in Equation \eqref{eqa2}, they are driven by the order book state. Although the Queue-reactive model cannot exactly be seen as a particular case of the framework presented in this paper (see Section \ref{specific}), most of the theoretical results shown in Sections 3 and 4 still hold (with some minor modifications in some assumptions). In particular, one can prove the diffusive limit of the reference price in the Queue-reactive model applying almost the same method as that used here.
\end{remark}

As soon as $p_{ref}$ changes, the value of $q_i$ switches immediately to the value of one of its neighbors (right if $p_{ref}$ increases, left if it decreases). As we keep only $K$ limits on each side, two boundary distributions $\pi_{-K}$ and $\pi_{K}$ are introduced for generating the new queue sizes at $Q_{-K}$ (when $p_{ref}$ decreases) and $Q_K$ (when $p_{ref}$ increases). To possibly incorporate external information, we moreover assume that with probability $\theta^{reinit}$, the LOB state vector $q(t)$ is redrawn from some distribution ($\pi^{inc}$ if $p_{ref}$ increases, $\pi^{dec}$ if $p_{ref}$ decreases) when $p_{ref}$ changes. As shown in \citet*{huang2013simulating}, models where price dynamics are purely endogenous, driven by order flows only, are usually not able to reproduce some of the important macroscopic features of prices, such as the volatility.  Thus the parameter $\theta^{reinit}$ can be understood as the percentage of price changes due to exogenous information, in which case market participants readjust very quickly their order flows around the new reference price, as if a new state of the LOB was drawn from its (invariant) distribution (ergodicity conditions are discussed in the next section). \\

For $q \in \Omega$, we write $q^+ = [q_{-K+1},...,q_{-1},q_{1},...,q_{K}]$, $q^- = [q_{-K},...,q_{-1},q_{1},...,q_{K-1}]$, $[q^+, l] = [q_{-K+1},...,q_{-1},q_{1},...,q_{K},l]$ and $[l,q^-] = [l,q_{-K},...,q_{-1},q_{1},...,q_{K-1}]$. Under the above assumptions, we have for $l\in \mathbb{Z}$ and $q,q',q'' \in \Omega$ such that $q'^+ \neq q^+$ and $q''^- \neq q^-$: 

\begin{eqnarray}
\mathcal{Q}_{(q,p),([q^+, l],p+\alpha)} &=& (1-\theta^{reinit})u(q)\pi_{K}(l) + \theta^{reinit}u(q)\pi^{inc}([q^+, l]) \nonumber\\
\mathcal{Q}_{(q,p),(q',p+\alpha)} &=& \theta^{reinit}u(q)\pi^{inc}(q') \nonumber\\
\mathcal{Q}_{(q,p),([l,q^-],p-\alpha)} &=& (1-\theta^{reinit})d(q)\pi_{-K}(l) + \theta^{reinit}d(q)\pi^{dec}([l,q^-]) \nonumber\\
\mathcal{Q}_{(q,p),(q'',p-\alpha)} &=& \theta^{reinit}d(q)\pi^{dec}(q''). \label{eqa3}
\end{eqnarray}

\subsubsection*{The infinitesimal generator matrix of $X(t)$}

Equations \eqref{eqa1}, \eqref{eqa2} and \eqref{eqa3} give a complete description of the infinitesimal generator matrix $\mathcal{Q}$ of the process $X(t)$, which is summarized in the following assumption. 

\begin{assumption}
Let $q, q', q'', \tilde{q} \in \Omega$, $p, \tilde{p} \in \alpha(0.5+\mathbb{Z})$, $n \in \mathbb{N}^{+}$, $l \in \mathbb{Z}$ be such that $q'^+ \neq q^+$ and $q''^- \neq q^-$. The process $X(t)$ is an irreducible Markov jump process with aperiodic embedded chain whose infinitesimal generator matrix $\mathcal{Q}$ is of the following form (with $2K$ functions $f_i,g_i: \Omega \times \mathbb{N}^{+} \to \mathbb{R}^{+}$ and two functions $u,d: \Omega \to \mathbb{R}^{+}$) : 
\begin{eqnarray}
\mathcal{Q}_{(q,p),(q+ne_i,p)} &=& f_i(q,n) \nonumber\\
\mathcal{Q}_{(q,p),(q-ne_i,p)} &=& g_i(q,n) \nonumber\\
\mathcal{Q}_{(q,p),([q^+, l],p+\alpha)} &=& (1-\theta^{reinit})u(q)\pi_{K}(l) + \theta^{reinit}u(q)\pi^{inc}([q^+, l]) \nonumber\\
\mathcal{Q}_{(q,p),(q',p+\alpha)} &=& \theta^{reinit}u(q)\pi^{inc}(q') \nonumber\\
\mathcal{Q}_{(q,p),([l,q^-],p-\alpha)} &=& (1-\theta^{reinit})d(q)\pi_{-K}(l) + \theta^{reinit}d(q)\pi^{dec}([l,q^-]) \nonumber\\
\mathcal{Q}_{(q,p),(q'', p-\alpha)} &=& \theta^{reinit}d(q)\pi^{dec}(q'') \nonumber\\
\mathcal{Q}_{(q,p),(q,p)} &=& -\sum_{(\tilde{q},\tilde{p}) \in \Omega \times \alpha(0.5+\mathbb{Z}), (\tilde{q},\tilde{p}) \neq (q,p)} \mathcal{Q}_{(q,p),(\tilde{q},\tilde{p})}, \nonumber \\
\mathcal{Q}_{(q,p),(\tilde{q},\tilde{p})} &=& 0, \emph{ otherwise.} \label{eqa4}
\end{eqnarray} \label{ass1}
\end{assumption} 

Note that under Assumption \eqref{ass1}, the dynamics of the process $X(t)$ is invariant under translations of the LOB center position: its infinitesimal matrix generator $\mathcal{Q}$ satisfies: $$\mathcal{Q}_{(q^1,p^1),(q^2,p^1+\beta)} = \mathcal{Q}_{(q^1,p^2),(q^2,p^2+\beta)},$$ for any $q^1,q^2 \in \Omega$, $p^1, p^2 \in \alpha(0.5+\mathbb{Z})$ and $\beta \in \alpha \mathbb{Z}$. One can also remark that in our framework, the order book state process $q(t)$ alone is also a continuous-time Markov jump process, whose ergodicity is discussed in the next section.

\subsection{Comparison with existing models}

The first major difference between our approach and the existing Markovian models in the literature is the introduction of state dependency in the order book dynamics. Most of the current order book models follow the ``Zero-intelligence'' framework, using Poisson flows for the processes of order arrivals. The Poisson assumption is clearly unrealistic, see for example the empirical results in \citet*{huang2013simulating}. In our framework, we propose to incorporate the strategic behaviors of market participants via a mean-field game approach, assuming their decisions depend on an underlying ``efficient" price $p_{ref}$ and on the LOB state vector $q(t)$. Note also that Equation \eqref{eqa1} allows us to have jumps of random size in the order book's shape, while a constant jump size is often assumed in the other existing models. \\


We also introduce a new method of LOB representation using one unique reference price $p_{ref}$. Most models use the best bid and best ask prices as two reference prices for indexing the buy and sell limits. In such models, these reference prices are directly determined by the order book state. In particular, changes in the order book state are immediately carried on the values of these prices. In our framework, $p_{ref}$ is not necessarily deduced from the order book state. Therefore, we can assume that changes in the order book state affect the value of $p_{ref}$ with some delay rate (the functions $u$ and $d$ introduced in Equation \eqref{eqa2}). Thanks to this original representation, we can naturally decompose the order book dynamics into two parts: a continuous-time multidimensional queuing system (Equation \eqref{eqa1}) and the dynamics of its center, that is the reference price (Equations \eqref{eqa2} and \eqref{eqa3}). Compared with the Queue-reactive model introduced in \citet*{huang2013simulating}, $p_{ref}$ is no longer constrained within the bid-ask spread. This desirable feature gives us the possibility of separating the exogenous and endogenous parts in the price dynamics by choosing appropriate price jump rate functions $u$ and $d$.  For example, with the functions $u$ and $d$ of Example \ref{ex1}, one can interpret $\theta_0$ as the exogenous part in the dynamics of $p_{ref}$ and $\theta_1$ as the intensity of the endogenous effects driving $p_{ref}$ towards the current mid price level.

\section{Ergodicity}\label{vue}

In this section, we discuss ergodicity properties in our framework. To do so, we make some additional assumptions on the functions $f_i$, $g_i$, $u$ and $d$.\\

Let
$$P_t(x,A) := \mathbb{P}[Y(t) \in A | Y(0) = x]$$
be the transition probability at time $t$ of a continuous or discrete-time Markov process $Y$ with state space $\mathcal{Y}$. In this work, we say that the process $Y$ is \emph{V-uniformly ergodic} if there exists a coercive\footnote{$|V(x)| \to +\infty$ as $||x|| \to \infty$.} function $V > 1$, an invariant distribution $\pi$, $r \in (0,1)$ and $R > \infty$ such that for any $x \in \mathcal{Y}$ and $t \in \mathbb{R}^+$ (or $\mathbb{N}^+$ in discrete-time),  
\begin{equation}\label{def_erg}
||P_t(x,.)-\pi(.)||_V \leq Rr^tV(x),\end{equation} where we write $||.||_V$ for the $V$-norm of a signed measure, see \citet*{meyn1993stability,Meyn2009}. In continuous time, the main idea to prove such property is to design an appropriate Lyapunov function $V$: $\mathcal{Y} \to (1,\infty)$, on which the following negative drift condition is satisfied for some $\gamma>0$ and $B > 0$: 
\begin{eqnarray}
\mathcal{Q}V(y) &:=& \sum_{y' \neq y} \mathcal{Q}_{yy'}[V(y')-V(y)]\nonumber \\
&\leq& -\gamma V(y) + B. \nonumber
\end{eqnarray}

Then by Theorem 6.1 in \citet*{meyn1993stability}, the Markov process $Y$ is non-explosive and V-uniformly ergodic. Furthermore, by Theorem 4.2 in \citet*{meyn1993stability} it is positive Harris recurrent. Note that the same kind of method is used in \citet*{abergel2011mathematical} in order to show ergodicity properties of Zero-intelligence models.\\

As mentioned in the introduction, the LOB's ergodicity implies here the existence of a unique invariant distribution for the state vector $q$. This is relevant for explaining the stylized empirical distribution of the LOB state. Mostly, as we will see in Section \ref{sec::scalinglimit}, the ergodicity analysis is the basis for proving the diffusive limit of the reference price process.  

\subsection{When $p_{ref}$ stays constant}

We first discuss the V-uniform ergodicity of the process $q(t)$ when assuming $u(q) = d(q) =0$ in Equation \eqref{eqa4}. Recall that the unused values of $f_i(q,n)$ and $g_i(q,n)$ in the definition of the queue dynamics, that is when $q \pm ne_i \notin \Omega$, are set to zero. With the convention $0/0=0$, we define
\begin{eqnarray}
f_i^*(q) &:=& \sum_n f_i(q,n) \nonumber \\
g_i^*(q) &:=& \sum_n g_i(q,n) \nonumber \\
l_i(q,n) &:=& \frac{f_i(q,n)}{f_i^*(q)} \nonumber \\
k_i(q,n) &:=& \frac{g_i(q,n)}{g_i^*(q)} \nonumber \\
G^{f,i,q}(z) &:=& \sum_{n=1}^\infty z^n l_i(q,n) \nonumber \\
G^{g,i,q}(z) &:=& \sum_{n=1}^\infty z^n k_i(q,n). \nonumber
\end{eqnarray}

Thus, when $f_i^*(q)>0$ (resp. $g_i^*(q)>0$), $l_i(q,.)$ (resp. $k_i(q,.)$) is a probability measure on $\mathbb{N^+}$ with moment-generating function $G^{f,i,q}(z)$ (resp. $G^{g,i,q}(z)$). We make the four following assumptions.

\begin{assumption}
For any order book state $q$ and any $i \geq i_{bestask}$, $g_i(q,n) = 0$ for any $n > q_i$ and for any order book state $q$ and any $i \leq i_{bestbid}$, $f_i(q,n) = 0$ for any $n > -q_i$. \label{ass6}
\end{assumption}
\begin{assumption}\label{zstar}
There exists $z^*>1$ such that for any $q$ and $i$, $G^{f,i,q}(z^*)<\infty$ and $G^{g,i,q}(z^*)<\infty$. Furthermore, there exists $L>0$ such that for any $i$,  $$\varlimsup_{z \to 1^+} \sup_{q} [f_i^*(q) G^{f,i,q}(z)\one_{i > i_{bestbid}} +  g_i^*(q) G^{g,i,q}(z)\one_{i < i_{bestask}}] < L.$$\label{ass7}
\end{assumption}
\begin{assumption}
There exist $r>0$ and $U>1$ such that 
\begin{eqnarray}
\varlimsup_{z\to 1^+} \sup_{(q,i):q_i>U, i \geq i_{bestask}}[f_i^*(q) - g_i^*(q)\frac{1-G^{g,i,q}(z^{-1})}{G^{f,i,q}(z)-1}] &<& -r \nonumber \\
\varlimsup_{z\to 1^+} \sup_{(q,i):q_i<-U,i \leq i_{bestbid}}[g_i^*(q) - f_i^*(q)\frac{1-G^{f,i,q}(z^{-1})}{G^{g,i,q}(z)-1}] &<& -r.
\end{eqnarray}\label{ass8}
\end{assumption}
\begin{assumption}
For any $z>1$,
\begin{eqnarray}
B_f(z) &:=&  \inf_{(q,i):q_i>U,i \geq i_{bestask}}(G^{f,i,q}(z)-1) > 0 \nonumber \\
B_g(z) &:=&  \inf_{(q,i):q_i<-U,i \leq i_{bestbid}}(G^{g,i,q}(z)-1) > 0. \nonumber
\end{eqnarray}\label{ass8p2}
\end{assumption}

To understand the practical meaning of these assumptions, let us consider the following example where the pure order book state jumps are assumed to have constant size equal to one. In such situation, the four assumptions above can be rewritten as follows and are much easier to interpret. 

\begin{example}\label{lobex2}
LOB model with constant order size. 
\begin{enumerate}
\item[$\bullet$] For $n \geq 2$, $f_i(q,n) = g_i(q,n) = 0$ for any $q \in \Omega$. 
\item[$\bullet$] There exists $L>0$ such that for any $i \in \{-K, ..., K\}$ and $q \in \Omega$,  $$f_i(q,1)\one_{i > i_{bestbid}} +g_i(q,1)\one_{i < i_{bestask}} < L.$$
\item[$\bullet$]There exist $r>0$ and $U>1$ such that 
\begin{eqnarray}
\sup_{(q,i):q_i>U,i\geq i_{bestask}} [f_i(q,1)-g_i(q,1)] &<& -r \nonumber \\
\sup_{(q,i):q_i<-U,i\leq i_{bestbid}} [g_i(q,1)-f_i(q,1)] &<& -r. \nonumber
\end{eqnarray}
\end{enumerate}
\end{example}

Basically, Assumption \ref{ass6} says that a bid/ask limit cannot become an ask/bid limit in a single queue update event, that is the queue size cannot revert its sign in a single jump\footnote{This assumption is not really mandatory but is realistic and technically quite convenient.}. From Example \ref{lobex2}, we see that Assumption \ref{ass7} essentially states that the total intensity of the order insertion processes in the bid and ask side remains uniformly bounded with respect to the state of the LOB. Assumption \ref{ass8}, which is the most important for the system's ergodicity, forces the individual queue sizes $|q_i|$ to decrease when they become larger than a certain threshold. From these assumptions, we obtain the following theorem proved in appendix for the Markov process $q(t)$. 

\begin{theorem}
When $u = d =0$, under Assumptions \ref{ass1}, \ref{ass6}, \ref{ass7}, \ref{ass8} and \ref{ass8p2}, the continuous-time Markov jump process $q(t)$ is non-explosive, V-uniformly ergodic and positive Harris recurrent.\label{theo1}
\end{theorem}

Consider now the embedded Markov chain $q(n)$ defined by $q(n)=q(J_n)$, where $J_n$ is the time of the $n$-th jump, and $q(J_n)$ the state of the LOB after this event. The study of the embedded Markov chain is an important step in order to obtain the diffusivity of the price process in our setting. We write 
$$a_i^*(q) = \frac{f_i^*(q)}{\sum_i[f_i^*(q)+g_i^*(q)]},~b_i^*(q) = \frac{g_i^*(q)}{\sum_i[f_i^*(q)+g_i^*(q)]},$$ for the proportions of queue size increases and decreases, and replace Assumption \ref{ass8} by the following one. 

\begin{assumption}
There exist $r>0$ and $U>1$ such that 
\begin{eqnarray}
\varlimsup_{z\to 1^+} \sup_{(q,i):q_i>U,i \geq i_{bestask}}[a_i^*(q) - b_i^*(q)\frac{1-G^{g,i,q}(z^{-1})}{G^{f,i,q}(z)-1}] &<& -r \nonumber \\
\varlimsup_{z\to 1^+} \sup_{(q,i):q_i<-U,i \leq i_{bestbid}}[b_i^*(q) - a_i^*(q)\frac{1-G^{f,i,q}(z^{-1})}{G^{g,i,q}(z)-1}] &<& -r. \nonumber
\end{eqnarray}\label{ass8add}
\end{assumption}

The following theorem is proved in appendix. 

\begin{theorem}
When $u = d =0$, under Assumptions \ref{ass1}, \ref{ass6}, \ref{ass7}, \ref{ass8p2} and \ref{ass8add}, the embedded discrete-time Markov chain $q(n)$ is V-uniformly ergodic and positive Harris recurrent. \label{theo2}
\end{theorem}

\subsection{General case}

We are now interested in the case where $u$ and $d$ are no longer fixed to $0$. Recall that $q(n)$ represents the state of the LOB after the $n$-th event and $p_{ref}(n)$ is the reference price (the center of the LOB) after the $n$-th event. We thus consider here the process of reference price increments $c(n)$ (since the reference price itself is of course not ergodic), defined as the reference price change at the $n$-th event: $$c(n) = p_{ref}(n) - p_{ref}(n-1),$$ and the embedded chain $Y(n) = (q(n),c(n))$, $n\in\mathbb{N}$, with $c(0)$ an artificial starting value. The process $Y(n)$ remains obviously Markovian. For some $z>1$, let 
$$V^z([q,c])=\sum_{i=-K,i\neq 0}^K{z^{|q_i|-U}}.$$ 
We make two additional assumptions for the general case.

\begin{assumption}\label{ass10}
There exist $z > 1$ and $L^{\pi}>0$ such that for $Q^{inc}$, $Q^{dec}$, $Q_K$, $Q_{-K}$ four random variables such that $Q^{inc}\sim \pi^{inc}$,  $Q^{dec}\sim \pi^{dec}$, $Q_{K}\sim \pi_{K}$ and $Q_{-K}\sim \pi_{-K}$:
$$\E[V^z([Q^{inc},c])]+\E[V^z([Q^{dec},c])]+\E[z^{|Q_{K}|- U}]+\E[z^{|Q_{-K}|- U}]\leq L^{\pi}.$$ 
\end{assumption}

\begin{assumption}
There exists a finite set $W \subset \Omega$ such that the upper bound of the proportion of reference price jumps in any order book state $q$ is smaller than one on $ \Omega /W$: 
$$\sup_{q \in \Omega /W} \frac{u(q)+d(q)}{\sum_i[f_i^*(q)+g_i^*(q)]+u(q)+d(q)}< 1.$$ \label{ass11} 
\end{assumption}

Assumption \ref{ass10} is technical and imposes some regularities on the four distributions $\pi^{inc}$, $\pi^{dec}$, $\pi_{K}$ and $\pi_{-K}$. Assumption \ref{ass11} ensures that a reference price change is not the only possible event except for a finite number of LOB states. Under these assumptions, we have the following theorem proved in appendix on the ergodicity of the embedded chain $Y(n)$.
\begin{theorem}
Under Assumptions \ref{ass1}, \ref{ass6}, \ref{ass7}, \ref{ass8p2}, \ref{ass8add}, \ref{ass10} and \ref{ass11}, the embedded discrete-time Markov chain $Y(n)=(q(n),c(n))$ is V-uniformly ergodic and positive Harris recurrent. \label{theo3}
\end{theorem}

\section{Scaling limits}\label{sec::scalinglimit}

We are now interested in the scaling limit of the reference price process. Let $J_i$ be the time of the $i$-th jump of the process.
Let $$N(t) = \inf\{n,~J_n \leq t\}$$ be the number of events until time $t$, with the convention $\inf\{\emptyset\}=0$. Let $Z(n)$ be the cumulative price change until the $n$-th event, that is $Z(0)=0$ and for $n\geq 1$: $$Z(n) = \sum_{i=1}^n c(i).$$ We have $$Z(N(t)) = p_{ref}(t) - p_{ref}(0).$$ Thus it represents the reference price at time $t$ recentered its starting value. We show in this section the diffusive behavior of $Z(N(t))$ as $n$ tends to infinity.\\

\noindent Consider again the embedded chain $Y(n) = (q(n),c(n))$. From Theorem \ref{theo3}, $Y(n)$ is V-uniformly ergodic towards an invariant distribution $\pi^*$. We have the following theorem for the rescaled price process in event time $\hat{S}^{(n)}(t):=\frac{Z({\lfloor nt \rfloor})}{\sqrt{n}}$.

\begin{theorem}
Under Assumptions \ref{ass1}, \ref{ass6}, \ref{ass7}, \ref{ass8p2}, \ref{ass8add}, \ref{ass10} and \ref{ass11}, if $\E_{\pi^*}[c(0)]=0$, then the series
\begin{equation}\label{defsigma}
\sigma^2 = \E_{\pi^*}[c_0^2] + 2\sum_{n=1}^\infty \E_{\pi^*}[c_0c_n],
\end{equation}
converges absolutely. Furthermore, if $Y(0) \sim \pi^*$, we have the following convergence in law in $D[0,\infty)$:  $$\hat{S}^{(n)}(t) \overset{n \to \infty}{\to} \sigma B(t),$$ where $B(t)$ is a standard Brownian motion. \label{theo41}
\end{theorem}

\begin{proof}
This theorem is a direct application of Theorem 19.1 in \citet*{billingsley2009convergence}. Indeed, the sequence $c_n$ is clearly stationary and ergodic in the sense of \citet*{billingsley2009convergence} (for example since it is stationary and mixing). Moreover, it has a finite second order moment and for all $n$, $\E_{\pi^*}[c_n] = \E_{\pi^*}[c_0] = 0$. 
\end{proof}


Theorem \ref{theo41} shows that in event time, the large scale limit of the reference price process is a Brownian motion. However, the most relevant question is that of the large scale limit of the reference price in calendar time. Thus we now consider the process $$\tilde{S}^{(n)}(t) = \frac{Z\big(N(nt)\big)}{\sqrt{n}}.$$ 

To prove the diffusivity of $\tilde{S}^{(n)}(t)$, we need a last assumption which is a bound of the expected value on the waiting time between two events.

\begin{assumption}
There exists some $m > 0$, such that
$$\inf_{q \in \Omega} \big\{\sum_i (f_i^*(q)+g_i^*(q)) + u(q)+d(q)\big\} > m.$$ \label{ass12}
\end{assumption}

Let $\tau_n$ be the inter-arrival time between the $n$-th and the $(n-1)$-th jumps of the Markov process $X$.
We then have the following theorem proved in appendix for the long term behavior of the reference price in calendar time. 

\begin{theorem} \label{theo_scaling}
Under Assumptions \ref{ass1}, \ref{ass6}, \ref{ass7}, \ref{ass8p2}, \ref{ass8add}, \ref{ass10}, \ref{ass11} and \ref{ass12}, the process $(q(n),c(n),\tau(n))$ is positive Harris recurrent. Furthermore, if $\E_{\pi^*}[c(0)]=0$ and $Y(0)\sim \pi^*$, then $$\tilde{S}^{(n)}(t) \overset{n \to \infty}{\to} \frac{\sigma}{\sqrt{\E_{\pi^{**}}[\tau(1)]}} B(t),$$ with $\pi^{**}$ is the invariant distribution of $(q(n),c(n),\tau(n))$ and $\sigma$ defined in \eqref{defsigma}.
\end{theorem}


Theorem \ref{theo_scaling} discusses the scaling limit of the underlying reference price. However, the difference between this price and the more usual $p_{bestbid}(t)$, $p_{bestask}(t)$ or $p_{mid}(t)$ being bounded by $2K$, the same result applies replacing the reference price by any of those prices. 

\section{Some specific models}\label{specific}

The Markovian setting proposed in this work allows us for a wide range of possibilities for modeling order book dynamics. The goal of this section is to give some natural and tractable examples of models, essentially already introduced in the literature, which can be seen as particular cases of our general framework. Together with the dynamics of the models, we provide sufficient conditions so that the assumptions made in the previous sections are satisfied in these specific models. Thus the ergodicity and diffusive scaling properties apply in all these models.

\subsection{Best bid/best ask Poisson model (\citet*{Cont2010})}

The basic idea of this first model, inspired by that introduced in \citet*{Cont2010}, is to use a constant spread size (fixed to 1 tick) and to consider only two limits in the order book. 

\begin{example}
Poisson model with $K = 1$.
\begin{enumerate}
\item[$\bullet$] We take $K=1$, $\theta^{reinit} = 1$ and assume that the functions $f_i$, $g_i$, $u$ and $d$ have the following forms, with $0<\lambda < \mu<\infty$: 
\begin{eqnarray}
f_1(q,n) &=& \lambda\one_{n= 1} \nonumber \\
g_1(q,n) &=& \mu\one_{q_1 > 0}\one_{n= 1}\nonumber \\
f_{-1}(q,n) &=& \mu\one_{q_{-1} < 0}\one_{n= 1}\nonumber \\
g_{-1}(q,n) &=& \lambda\one_{n= 1} \nonumber \\
u(q) &=& \theta \one_{q_1 = 0} \nonumber \\
d(q) &=& \theta \one_{q_{-1} = 0}. \nonumber
\end{eqnarray}
\item[$\bullet$] $\pi^{inc}$ and $\pi^{dec}$ satisfy Assumption \ref{ass10} and\\

for any $q_{-1} > 0$, $q_1 \in \mathbb{Z}$, 
$$\pi^{inc}(q_{-1},q_1)= \pi^{dec}(q_{-1},q_1) = 0,$$
for any $q_1 < 0$, $q_{-1} \in \mathbb{Z}$, 
$$\pi^{inc}(q_{-1},q_1) = \pi^{dec}(q_{-1},q_1) = 0. $$ 
\end{enumerate}
\end{example}

Note that here, the boundary distributions $\pi_{K}$ and $\pi_{-K}$ are no longer needed, since the order book reinitialization probability $\theta^{reinit}$ is set to one.\\

In this model, the role of $p_{ref}$ is very close to that of $p_{mid}$, which splits the order book into two parts: the bid side ($Q_{-1}$) and the ask side ($Q_1$). The limit order insertion, cancellation and market order insertion processes are assumed to be independent Poisson processes. The size of these orders is assumed to be constant and $p_{ref}$ jumps with rate $0$ when none of the queues $Q_{\pm 1}$ is empty, with rate $\theta$ to the right side when $Q_{1}$ is empty, with rate $\theta$ to the left side when $Q_{-1}$ is empty. When the value of $\theta$ is very large, the price jump is almost instantaneous as soon as one of the two queues becomes empty. In that case, this model becomes very close to that proposed in \citet*{Cont2010}, where an infinite rate is used (note that the convergence of the rescaled price process can still be proved with some minor modifications in such case of infinite jump rate). 

\subsection{Poisson model with $K > 1$}\label{sec:PoissonK2}

It is natural to try to extend the previous Poisson model in order to include more queues in the order book and to allow for a spread size different from one tick. In such model, the role of $p_{ref}$ is slightly different since it is not necessarily the mid price. Again, $p_{ref}$ can be understood here as the underlying efficient price that determines the order arrival intensities at different price levels. Now buy/sell limit orders can be inserted both on the right side and on the left side of  $p_{ref}$.

\begin{example}
Poisson model with $K > 1$. 
\begin{enumerate}
\item[$\bullet$] The functions $f_i$, $g_i$, $u$ and $d$ have the following forms, for $i = -K, ..., K$: 
\begin{eqnarray}
f_i(q,1) &=& \lambda_i \one_{i > i_{bestbid}(q)} + \gamma_i\one_{i = i_{bestbid}(q)} \nonumber \\
&& + \mu_i \one_{i \leq i_{bestbid}(q)}\one_{q_i<0} \nonumber \\
g_i(q,1) &=& \lambda_{-i} \one_{i < i_{bestask}(q)} + \gamma_{-i}\one_{i = i_{bestask}(q)} \nonumber \\
&& + \mu_{-i}  \one_{i \geq i_{bestask}(q)}\one_{q_i>0} \nonumber \\
f_i(q,n) &=& 0, \mbox{for n} > 1 \nonumber \\
g_i(q,n) &=& 0, \mbox{for n} > 1 \nonumber \\
u(q) &=& \theta_{i_{bestask}(q)} \nonumber \\
d(q) &=& \theta_{-i_{bestbid}(q)}. \nonumber
\end{eqnarray}
\item[$\bullet$] $\pi^K$, $\pi^{-K}$, $\pi^{inc}$ and $\pi^{dec}$ satisfy Assumption \ref{ass10}. 
\item[$\bullet$] For any $i$,$j \in \{-K,...,K \}$, $i< j$, we have 
\begin{eqnarray}
\mu_{-i} > \lambda_i>0 \nonumber \\
0\leq\theta_{i}\leq \theta_{j}. \nonumber
\end{eqnarray} 
\end{enumerate}
\end{example}

Limit orders, cancellations and market orders (which consume the quantities at the best offer limits) are modeled by independent Poisson processes with different intensities $\lambda_i^{buy/sell}$, $\mu_i^{buy/sell}$ and $\gamma_i^{buy/sell}$, depending on the distance from their target price to $p_{ref}$. We assume bid-ask symmetry in this model, that is $\lambda_i^{buy} = \lambda_{-i}^{sell}$, $\mu_i^{buy} = \mu_{-i}^{sell}$ and $\gamma_i^{buy} = \gamma_i^{sell}$, thus we omit the index buy/sell in the above equations. Remark that the intensity of the buy/sell market order flow at the best limit $\gamma_i$ is a function of $i$, that is the position of the best limit with respect to the reference price. This allows us to model the fact that market participants have different behaviors towards the best limit, depending on their evaluation of the reference price.\\

\noindent The reference price jump dynamics is modeled by a sequence of increasing rates $\theta_i,i \in \{-K,...,K\}$. This means that the larger the index of the best ask queue, the larger the probability of $p_{ref}$ to increase and the smaller the index of the best bid queue, the larger the probability of $p_{ref}$ to decrease. Note that in this model, we no longer assume any specific value for the reinitialization probability $\theta^{reinit}$ and use Assumption \ref{ass10} to impose some properties on the boundary distributions $\pi^K$ and $\pi^{-K}$ and the initialization distributions $\pi^{inc}$ and $\pi^{dec}$. 

\subsection{Zero-intelligence model}

We now present a different way of extending the Poisson model with $K=1$ in order to include more queues in the order book. This modeling approach where two reference prices are used is called Zero-intelligence model and is introduced in \citet*{smith2003statistical}. It is also considered in \citet*{abergel2011mathematical} and is the basis of \citet*{cont2010stochastic}. We define $\phi(i,j)$ as the absolute distance (in number of ticks) between the queue $Q_i$ and $Q_j$ ($\phi(i,j) = |i-j|$) and make the three following assumptions.  

\begin{example}
Zero-intelligence model.
\begin{enumerate}
\item[$\bullet$] The functions $f_i$, $g_i$, $u$ and $d$ have the following forms, for $i = -K, ..., K$: 
\begin{eqnarray}
f_i(q,1) &=& \lambda_{\phi(i,i_{bestbid}(q))} \one_{i > i_{bestbid}(q)} + \gamma \one_{i = i_{bestbid}(q)} + |q_i| \mu_{\phi(i,i_{bestask}(q))} \one_{i \leq i_{bestbid}(q)} \nonumber \\
g_i(q,1) &=& \lambda_{\phi(i,i_{bestask}(q))}\one_{i < i_{bestask}(q)} + \gamma \one_{i = i_{bestask}(q)} + |q_i| \mu_{\phi(i,i_{bestbid}(q))} \one_{i \geq i_{bestask}(q)} \nonumber \\
f_i(q,n) &=& 0, \mbox{for n} > 1 \nonumber \\
g_i(q,n) &=& 0, \mbox{for n} > 1 \nonumber \\
u(q) &=& \theta_{i_{bestask}(q)} \nonumber \\
d(q) &=& \theta_{-i_{bestbid}(q)}. \nonumber
\end{eqnarray}
\item[$\bullet$] $\pi^K$, $\pi^{-K}$, $\pi^{inc}$ and $\pi^{dec}$ satisfy Assumption \ref{ass10}. 
\item[$\bullet$] For any $i$,$j \in \{-K,...,K \}$, $\lambda_{\phi(i,j)}>0$ and $\mu_{\phi(i,j)}>0$.
\item[$\bullet$] For any $i$,$j \in \{-K,...,K \}$, $i< j$, we have 
\begin{eqnarray}
0\leq\theta_{i}\leq \theta_{j}. \nonumber
\end{eqnarray} 
\end{enumerate}
\end{example}

In this model, $p_{ref}$ is no longer an underlying efficient price determining the order arrival intensities. These intensities now depend on the positions of two different prices: $p_{bestbid}$ and $p_{bestask}$. Limit orders, cancellations and market orders are still described by independent Poisson processes. Buy/sell limit orders are inserted in the queues to the left/right side of the best ask/best bid price, with intensities depending on the distance between their price level and the best ask/best bid price ($\lambda_{\phi(i,i_{bestask/bestbid})}$); cancellations of buy/sell orders are sent to the queues on the left/right side of the best ask/best bid price, with intensities being linear functions of the queue sizes ($|q_i| \mu_{\phi(i,i_{bestask/bestbid})}$); market buy/sell orders are sent to the best ask/best bid queue, with intensity $\gamma$. The reference price $p_{ref}$ now provides the center of the $2K$ dimensional moving frame representing the LOB's state and the same modeling approach as in Section \ref{sec:PoissonK2} is used for its dynamics.

\subsection{Queue-reactive model (\citet*{huang2013simulating})}\label{sec::queuereactive}

In \citet*{huang2013simulating}, the Queue-reactive model for order books is introduced. This model takes into account the influence of the order book's state in determining the order arrival intensities (in a much more general way than considering only the position of the best bid and best ask queues). The Queue-reactive model assumes that no buy/sell limit order can be inserted on the right/left side of $p_{ref}$ and uses the following assumption instead of Equation \eqref{eqa2} for the dynamics of the jumps of $p_{ref}$.

\begin{assumption}
Whenever $p_{mid}$ increases (resp. decreases), $p_{ref}$ increases (resp. decreases) by $\alpha$ with probability $\theta$, provided that $q_{1}=0$ (resp. $q_{-1}=0$) at that moment. Therefore, changes in the value of $p_{ref}$ are possibly triggered by one of the three following events:
\begin{itemize}
\item The insertion of a buy (resp. sell) limit order within the bid-ask spread while $Q_{1}$ (resp. $Q_{-1}$ is empty). 
\item The cancellation of the last limit order at one of the best offer queues.
\item A market order which consumes the last limit order at one of the best offer queues.
\end{itemize}\label{assoriginqueuereactive}
\end{assumption}

With some minor modifications in the proof of Theorem \ref{theo_scaling}, one can prove that the scaling limit of $p_{ref}$ in the Queue-reactive model is a Brownian motion. As explained above, in this model, changes of $p_{ref}$ are triggered by events that modify the mid price. Here we propose a slightly modified version of the Queue-reactive model\footnote{Model I in \citet*{huang2013simulating} is used to describe the queue dynamics during constant reference price periods.} by considering the following four assumptions (note that the state space $\Omega$ is reduced in that case to $\Omega^*:= \{q \in \Omega, q_i\one_{i<0} \leq 0, q_i\one_{i>0} \geq 0\}$). 

\begin{example}
Queue-reactive type model.
\begin{enumerate}
\item[$\bullet$] The functions $f_i$, $g_i$, $u$ and $d$ have the following forms, for $i = -K, ..., K$: 
\begin{eqnarray}
f_i(q,1) &=& \lambda_{|i|}(q_i)\one_{i > 0} + \mu_{|i|}(-q_i) \one_{i < 0} \nonumber \\
g_i(q,1) &=& \lambda_{|i|}(-q_i)\one_{i < 0} + \mu_{|i|}(q_i) \one_{i > 0} \nonumber \\
f_i(q,n) &=& 0, \mbox{for n} > 1 \nonumber \\
g_i(q,n) &=& 0, \mbox{for n} > 1 \nonumber \\
u(q) &=& \theta \one_{q_1 = 0} \nonumber \\
d(q) &=& \theta \one_{q_{-1} = 0}, \nonumber
\end{eqnarray} 
with $\lambda_{|i|}$ and $\mu_{|i|}$ non-negative functions defined on $\mathbb{N}$, with $\mu_{|i|}(0) = 0$.
\item[$\bullet$]  We have 
\begin{eqnarray}
\sup_{i \in \{1,...,K\},q_i \in \mathbb{N}} (\lambda_i(q_i)) < L < \infty. \nonumber
\end{eqnarray} 
\item[$\bullet$] There exist $r>0$ and $U>1$ such that for any $q_i > U$ and any $i \in \{1,...,K\}$: 
\begin{eqnarray}
\lambda_{i}(q_i)- \mu_{i}(q_i)&<& -r. \nonumber
\end{eqnarray}
\item[$\bullet$] There exists $m'>0$ such that for any $i \in \{1,...,N\}$:
\begin{eqnarray}
\inf_{q_i \in \mathbb{N}}[\lambda_{i}(q_i)+\mu_{i}(q_i)] &>& m'. \nonumber
\end{eqnarray}
\end{enumerate}
\end{example}

Compared with Assumption \ref{assoriginqueuereactive}, changes in $p_{ref}$ are now driven by the relative position of the mid price in the current order book state. Nevertheless, we can see that the two approaches are quite similar. In this model, $p_{ref}$ always stays between $p_{bestask}$ and $p_{bestbid}$ (since $\mu_{|i|}(0) = 0$ implies that the queue sizes on the left/right side of $p_{ref}$ never become positive/negative). Such model gives us a much larger choice on the intensity functions $\lambda_{|i|}$ and $\mu_{|i|}$ than when assuming Poisson flows. Furthermore, with enough data points, these functions can be estimated in a non-parametric way, as done in \citet*{huang2013simulating}. Finally the state-dependent approach provides us very interesting insights about the way the order book state influences market participants decisions and the mechanism making the empirical distribution of the order book arise from these decisions, see \citet*{huang2013simulating}.

\section{Conclusion}

In this work, we extend the order book modeling approach proposed in \citet*{huang2013simulating} to a more general Markovian framework, allowing to take into account most relevant features of LOB dynamics such as: 
\begin{enumerate}
\item[$\bullet$] Dependencies between the order arrival processes and the LOB state.
\item[$\bullet$] Endogenous movements of the underlying efficient price and influence of the LOB state on its dynamics. 
\item[$\bullet$] Exogenous movements of the underlying efficient price.
\item[$\bullet$] Randomness in the size of the orders.
\end{enumerate}

The ergodicity of the LOB system and the diffusive limit of the rescaled price process are established under general assumptions. Finally, to illustrate the usefulness and the relevance of our approach, several examples of classical models which can be seen as particular cases of our general framework are presented.\\

\noindent To get a fully satisfying model, a last step would probably be to allow for a non-Markovian component in the market order flow (since the Markov assumption is probably quite reasonable for the limit order and cancellation flows). This can for example be done using self-exciting processes, as in \citet*{abergel2015}. However, except for very specific cases (exponential Hawkes processes for example), adding such non-Markovian component would certainly require revising completely the mathematical approach to the model. 

\section*{Ackowledgements}
We thank Fr\'ed\'eric Abergel for helpful discussions.

\section{Appendix}

\subsection{Proof of Theorem \ref{theo1}}

Let us denote by $\tilde{\mathcal{Q}}$ the infinitesimal generator matrix of $q(t)$ when $u = d=0$. The infinite matrix $\tilde{\mathcal{Q}}$ has the following form: 
\begin{eqnarray}
\tilde{\mathcal{Q}}_{q,q+ne_i} &=& f_i(q,n) \nonumber\\
\tilde{\mathcal{Q}}_{q,q -ne_i} &=& g_i(q,n) \nonumber\\
\tilde{\mathcal{Q}}_{q,q} &=& -\sum_{q \in \Omega, q' \neq q} \tilde{\mathcal{Q}}_{q,q'} \nonumber\\
\tilde{\mathcal{Q}}_{q,q'} &=& 0, \mbox{otherwise.} \nonumber
\end{eqnarray}
For some $1<z\leq z^*$ (recall that $z^*$ is defined in Assumption \ref{zstar}), let us consider the function $$V(q)=\sum_{i=-K,i\neq 0}^K{z^{|q_i|-U}}.$$ Since $q_i \geq 0$ for $i \geq i_{bestask}$, $q_i \leq 0$ for $i \leq i_{bestbid}$ and $q_i = 0$ for $i \in (i_{bestbid},i_{bestask})$, we have 
\begin{eqnarray}
\tilde{\mathcal{Q}}V(q) &:=& \sum_{q' \neq q} \tilde{\mathcal{Q}}_{qq'}[V(q')-V(q)]\nonumber \\
&=& \sum_{i \leq i_{bestbid}} \sum_{n=1}^\infty [(z^{-q_i+n-U} - z^{-q_i-U})g_i(q,n) + (z^{|q_i+n|-U} - z^{-q_i-U})f_i(q,n)] \nonumber \\
&& +  \sum_{i \geq i_{bestask}} \sum_{n=1}^\infty [(z^{q_i+n-U} - z^{q_i-U})f_i(q,n) + (z^{|q_i-n|-U} - z^{q_i-U})g_i(q,n)] \nonumber \\
&& + \sum_{i \in (i_{bestbid}, i_{bestask})} \sum_{n=1}^\infty [(z^{n-U} - z^{-U})f_i(q,n) + (z^{n-U}- z^{-U})g_i(q,n)]. \nonumber
\end{eqnarray}

Then by Assumption \ref{ass6},
\begin{eqnarray}
\tilde{\mathcal{Q}}V(q) &=& \sum_{i \leq i_{bestbid}} \sum_{n=1}^\infty [(z^{-q_i+n-U} - z^{-q_i-U})g_i(q,n) + (z^{-q_i-n-U} - z^{-q_i-U})f_i(q,n)] \nonumber \\
&& +  \sum_{i \geq i_{bestask}} \sum_{n=1}^\infty [(z^{q_i+n-U} - z^{q_i-U})f_i(q,n) + (z^{q_i-n-U} - z^{q_i-U})g_i(q,n)] \nonumber \\
&& + \sum_{i \in (i_{bestbid}, i_{bestask})} \sum_{n=1}^\infty [(z^{n-U} - z^{-U})f_i(q,n) + (z^{n-U}- z^{-U})g_i(q,n)] \nonumber \\
&=& \sum_{i \leq i_{bestbid}}  z^{-q_i-U} \sum_{n=1}^\infty (z^n  - 1)[g_i(q,n) - \frac{f_i(q,n)}{z^n}] \nonumber \\
&& + \sum_{i \geq i_{bestask}}  z^{q_i-U} \sum_{n=1}^\infty (z^n  - 1)[f_i(q,n) - \frac{g_i(q,n)}{z^n}] \nonumber \\
&& + \sum_{i \in (i_{bestbid}, i_{bestask})}  z^{-U} \sum_{n=1}^\infty (z^n  - 1)[f_i(q,n) + g_i(q,n)]. \nonumber
\end{eqnarray}

Using the definition of $G^{f,i,q}(z)$, $G^{g,i,q}(z)$, $f_i^*(q)$ and $g_i^*(q)$, we get
\begin{eqnarray}
\tilde{\mathcal{Q}}V(q) &=& \sum_{i \leq i_{bestbid}}  z^{-q_i-U} [g_i^*(q) (G^{g,i,q}(z)-1) - f_i^*(q) (1-G^{f,i,q}(z^{-1}))] \nonumber \\
&& + \sum_{i \geq i_{bestask}}  z^{q_i-U} [f_i^*(q) (G^{f,i,q}(z)-1) - g_i^*(q) (1-G^{g,i,q}(z^{-1}))] \nonumber \\
&& + \sum_{i \in (i_{bestbid}, i_{bestask})}  z^{-U} [f_i^*(q)(G^{f,i,q}(z)-1)+ g_i^*(q)(G^{g,i,q}(z)-1)]. \nonumber
\end{eqnarray}

Moreover, since for $z > 1$ we have $G^{f/g,i,q}(z^{-1}) < 1$ and $G^{f/g,i,q}(z) > 1$, we obtain 
\begin{eqnarray}
\tilde{\mathcal{Q}}V(q) &\leq& \sum_{i: i \leq i_{bestbid}, q_i < -U} z^{-q_i-U} (G^{g,i,q}(z)-1)\big(g_i^*(q) - f_i^*(q) \frac{1-G^{f,i,q}(z^{-1})}{G^{g,i,q}(z)-1}\big) \nonumber \\
&& + \sum_{i: i \geq i_{bestask}, q_i > U} z^{q_i-U} (G^{f,i,q}(z)-1)\big(f_i^*(q) - g_i^*(q) \frac{1-G^{g,i,q}(z^{-1})}{G^{f,i,q}(z)-1}\big)\nonumber \\
&& + \sum_{i: i \leq i_{bestbid},q_i \geq -U} g_i^*(q) G^{g,i,q}(z) + \sum_{i: i \geq i_{bestask},q_i \leq U} f_i^*(q) G^{f,i,q}(z) \nonumber \\
&& + \sum_{i \in (i_{bestbid}, i_{bestask})} [f_i^*(q)G^{f,i,q}(z)+ g_i^*(q)G^{g,i,q}(z)] \nonumber \\
&\leq& \sum_{i: i \leq i_{bestbid}, q_i < -U} z^{-q_i-U} (G^{g,i,q}(z)-1)\big(g_i^*(q) - f_i^*(q) \frac{1-G^{f,i,q}(z^{-1})}{G^{g,i,q}(z)-1}\big) \nonumber \\
&& + \sum_{i: i \geq i_{bestask}, q_i > U} z^{q_i-U} (G^{f,i,q}(z)-1)\big(f_i^*(q) - g_i^*(q) \frac{1-G^{g,i,q}(z^{-1})}{G^{f,i,q}(z)-1}\big) \nonumber \\
&& + \sum_{i} [f_i^*(q) G^{f,i,q}(z)\one_{i > i_{bestbid}} +  g_i^*(q) G^{g,i,q}(z)\one_{i < i_{bestask}}]. \nonumber
\end{eqnarray}

Now note that $$\sup_q [f_i^*(q) G^{f,i,q}(z)\one_{i > i_{bestbid}} +  g_i^*(q) G^{g,i,q}(z)\one_{i < i_{bestask}}]$$ is an increasing function of $z$. Thus by Assumption \ref{ass7}, we can find $z'>1$ such that for any $z \leq z'$, $q \in \Omega$ and $i \in [-K,...,K]$: $$ \big(f_i^*(q) G^{f,i,q}(z)\one_{i > i_{bestbid}} +  g_i^*(q) G^{g,i,q}(z)\one_{i < i_{bestask}}\big) < L.$$ From Assumption \ref{ass8}, we see that we can find some $\tilde{z}$ with
$1<\tilde{z}\leq z'$ such that for any $(q,i), q_i < -U,i \leq i_{bestbid}$, for any $1<z\leq \tilde{z}$:
\begin{equation}
g_i^*(q) - f_i^*(q) \frac{1-G^{f,i,q}(z^{-1})}{G^{g,i,q}(z)-1} < -r,
\end{equation}
and for any $(q,i),q_i > U,i \geq i_{bestask}$, for any $1<z\leq \tilde{z}$:
\begin{equation}
f_i^*(q) - g_i^*(q) \frac{1-G^{g,i,q}(z^{-1})}{G^{f,i,q}(z)-1} < -r.
\end{equation}
Thus taking $z$ in the definition of the function $V$ satisfying $2Kz^{-U}>1$ (so that the function $V$ is coercive) and $1<z\leq \tilde{z}$, we obtain 
\begin{eqnarray}
\tilde{\mathcal{Q}}V(q) &\leq& -r \sum_{i: i \leq i_{bestbid}, q_i < -U} z^{-q_i-U} (G^{g,i,q}(z)-1) \nonumber \\
&& -r \sum_{i: i \geq i_{bestask}, q_i > U} z^{q_i-U} (G^{f,i,q}(z)-1) + 2KL. \nonumber \\
 &\leq& -rB_g(z) \sum_{i: i \leq i_{bestbid}, q_i < -U} z^{-q_i-U} \nonumber \\
&& -rB_f(z) \sum_{i: i \geq i_{bestask}, q_i > U} z^{q_i-U} + 2KL. \nonumber
\end{eqnarray}

By Assumption \ref{ass8p2}, $B:=\min (B_g(z),B_f(z))> 0$. Therefore we get 

\begin{eqnarray}
\tilde{\mathcal{Q}}V(q) &\leq& -rB \sum_{i: |q_i| > U}  z^{|q_i|-U} + 2KL  \nonumber \\
&\leq& -rB \sum_i  z^{|q_i|-U} + 2K(L+rB)   \nonumber \\
&=& -rBV(q) + 2K(L+rB). \nonumber
\end{eqnarray}
Finally, remark that in our setting, any compact set included in $\Omega$ is finite. A singleton being a petite set and a finite union of singletons remaining a petite set, see Proposition 5.5.5 in \citet*{Meyn2009}, we get that all the compact sets are petite. Therefore by Theorem 6.1 in \citet*{meyn1993stability}, $q(t)$ is non-explosive and V-uniformly ergodic. Furthermore, by Theorem 4.2 in \citet*{meyn1993stability} it is positive Harris recurrent.

\subsection{Proof of Theorem \ref{theo2}}

For some $1<z\leq z^*$, set again $$V(q)=\sum_{i=-K,i\neq 0}^K{z^{|q_i|-U}}.$$ We write $\tilde{\mathcal{P}}_{q,q'}$ the transition probability from $q$ to $q'$. In the same way as in the preceding proof, we have 
\begin{eqnarray}
\Delta V(q) &:=& \sum_{q' \in \Omega} \tilde{\mathcal{P}}_{q,q'}(V(q')-V(q)) \nonumber \\
&=& \sum_{i \leq i_{bestbid}}  z^{-q_i-U} [b_i^*(q) (G^{g,i,q}(z)-1) - a_i^*(q) (1-G^{f,i,q}(z^{-1}))] \nonumber \\
&& + \sum_{i \geq i_{bestask}}  z^{q_i-U} [a_i^*(q) (G^{f,i,q}(z)-1) - b_i^*(q) (1-G^{g,i,q}(z^{-1}))] \nonumber \\
&& + \sum_{i \in (i_{bestbid}, i_{bestask})}  z^{-U} [a_i^*(q)(G^{f,i,q}(z)-1)+ b_i^*(q)(G^{g,i,q}(z)-1)]. \nonumber
\end{eqnarray}

Following the same method as in the proof of Theorem \ref{theo1}, we can easily find $1<z'\leq z^*$ and $B>0$ such that taking $z=z'$ in the definition of $V$, it is coercive and we get 
\begin{eqnarray}
\Delta V(q) &\leq& -rV(q) + B. \nonumber
\end{eqnarray}

Now define the set $C := \{q,~rV(q) \leq 2B \}$, C is obviously a finite set and is therefore petite. Furthermore, we have 
\begin{equation}
\Delta V(q) \leq -\frac{r}{2}V(q) + B\one_{q \in C}. \nonumber
\end{equation}
Thus by Theorem 16.1.2 in \citet*{Meyn2009}, $q(n)$ is V-uniformly ergodic.\\

Eventually, the fact that the chain is positive Harris recurrent is deduced from Theorem 9.1.8 together with Theorem 15.0.1 in \citet*{Meyn2009}. 

\subsection{Proof of Theorem \ref{theo3}}

For ease of notation, we write $V$ instead of $V^z$. Let
$Q^{inc}$, $Q^{dec}$, $Q_K$, $Q_{-K}$ be four random variables such that $Q^{inc}\sim \pi^{inc}$,  $Q^{dec}\sim \pi^{dec}$, $Q_{K}\sim \pi_{K}$ and $Q_{-K}\sim \pi_{-K}$. We define 
\begin{eqnarray}
u^*(q) &=& \frac{u(q)}{\sum_i[f_i^*(q)+g_i^*(q)]+u(q)+d(q)}\nonumber \\
d^*(q) &=& \frac{d(q)}{\sum_i[f_i^*(q)+g_i^*(q)]+u(q)+d(q)}\nonumber \\
n^*(q) &=& \frac{\sum_i[f_i^*(q)+g_i^*(q)]}{\sum_i[f_i^*(q)+g_i^*(q)]+u(q)+d(q)}\nonumber \\
\E_{K} &=& \E[z^{Q_{K} - U}]\nonumber \\
\E_{-K} &=& \E[z^{Q_{-K} - U}]\nonumber \\
\E_{\pi^{inc}} &=& \E[V([Q^{inc},c])]\nonumber \\
\E_{\pi^{dec}} &=& \E[V([Q^{dec},c])].\nonumber
\end{eqnarray}
Remarking that $V([q,c])$ does not depend on $c$, we write from now on $V(q)$ instead of $V([q,c])$. Moreover, we set $\mathcal{P}_{[q,c],[q',c']}$ as the transition probability from state $[q,c]$ to state $[q',c']$ and $\tilde{\mathcal{P}}_{q,q'}$ as the transition matrix of the embedded chain $q(n)$ when $u=d=0$. Using the form of the infinitesimal generator $\mathcal{Q}$, we deduce
\begin{eqnarray}
\Delta V([q,c]) &:=& \sum_{(q',c') \in \Omega \times \{-\alpha,\alpha\}} \mathcal{P}_{[q,c],[q',c']}(V(q')-V(q)) \nonumber \\
&=& n^*(q)  \sum_{q' \in \Omega} \tilde{\mathcal{P}}_{q,q'}(V(q')-V(q)) \nonumber \\
&& + u^*(q) \big[(1-\theta^{reinit})(\E_{K} - z^{|q_{-K}|-U}) + \theta^{reinit}(\E_{\pi^{inc}} - V(q))\big] \nonumber \\
&& + d^*(q) \big[(1-\theta^{reinit})(\E_{-K} - z^{|q_{K}|-U}) + \theta^{reinit}(\E_{\pi^{dec}} - V(q))\big]. \nonumber
\end{eqnarray}

By Assumption \ref{ass10}, we have
\begin{eqnarray}
\Delta V([q,c]) &\leq& n^*(q)  \sum_{q' \in \Omega} \tilde{\mathcal{P}}_{q,q'}(V(q')-V(q)) + 2L^K + 2L^\pi. \nonumber
\end{eqnarray}
Then as in the proof of Theorem \ref{theo2}, we can find $1<z'\leq z^*$ and $B>0$ such that taking $z=z'$ in the definition of $V$, it is coercive and we get 
\begin{eqnarray}
\Delta V(q) &\leq& -n^*(q)rV(q) + Bn^*(q)+2L^K+ 2L^\pi. \nonumber
\end{eqnarray}
Moreover, Assumption \ref{ass11} ensures that for any $q$ except those belonging to the finite set $W$, $n^*(q)>M$ with $M\in(0,1]$. Consequently, 
\begin{eqnarray}
\Delta V(q) &\leq& -MrV(q) +B+2L^K+ 2L^\pi+\nu\one_{[q,c]\in W\times\{-\alpha,\alpha\}}, \nonumber
\end{eqnarray}
with $\nu>0$. Now define the set $C := \{[q,c],~MrV(q) \leq 2B+4L^K+ 4L^\pi \}$. Being finite, $C\cup W$ is a petite set and we have 

\begin{equation}
\Delta V([q,c]) \leq -\frac{Mr}{2}V([q,c]) + (B+2L^K+ 2L^\pi+\nu)\one_{[q,c] \in C\cup W}. \nonumber
\end{equation}
Hence from Theorem 16.1.2 in \citet*{Meyn2009}, $Y(n)$ is V-uniformly ergodic.\\

Eventually, the fact that the chain is positive Harris recurrent is deduced from Theorem 9.1.8 together with Theorem 15.0.1 in \citet*{Meyn2009}.

\subsection{Proof of Theorem \ref{theo_scaling}}

\subsubsection{Preliminary lemma}

We start with the following preliminary lemma.

\begin{lemma}\label{lem1}
For the Markov chain $(q(n),c(n),\tau(n))$, the Cartesian product of any finite set included in $\Omega \times \{-\alpha, \alpha\}$ and $\mathbb{R}^+$ is petite. 
\end{lemma}
\begin{proof}
We first show that for any $q \in \Omega$ and $c \in \{-\alpha, \alpha\}$, the set ${q} \times {c} \times \mathbb{R}^+$ is petite (actually small). We define the measure $\nu_{q,c,\tau}$ which is so that for any $B \in \mathcal{B}(\Omega \times  \{-\alpha, \alpha\} \times \mathbb{R}^+)$, $\nu_{q,c,\tau}(B)$ is the transition probability from $[q,c,\tau]$ to $B$ in a single step: $$\nu_{q,c,\tau}(B) = \mathcal{P}_{[q,c,\tau],B}.$$
 
Recall that in our framework, the transition probabilities from $(q(n),c(n),\tau(n))$ depend only on the value of $q(n)$. So we can write $\nu_{q,c,\tau}(B)$ as $\nu_{q}(B)$. In the sense of Equation (5.43) in \citet*{Meyn2009}, the transition probability $\mathcal{P}_{[q,r,\tau], B}$ can be seen as a sampling kernel for the Markov chain $(q(n),c(n),\tau(n))$, using the Dirac measure at point $1$ on $\mathbb{Z}^+$ as sampling measure. Moreover, for any $\tau \in \mathbb{R}^+$ and any $B \in \mathcal{B}(\Omega \times  \{-\alpha, \alpha\} \times \mathbb{R}^+)$, we have
$$\mathcal{P}_{[q,c,\tau], B}\geq \nu_q(B).$$

Therefore the set ${q} \times {c} \times \mathbb{R}^+$ is petite. Then, as the union of two petite sets remains petite, see Proposition 5.5.5 in \citet*{Meyn2009}, we have the result.
\end{proof}

\subsubsection{A law of large numbers}

In the next proposition, we give a law of large numbers for the inter-arrival times, which is the key element to establish the diffusive behavior of the price in calendar time. Within the proof of this proposition, we show that $(q(n),c(n),\tau(n))$ is positive Harris recurrent.

\begin{proposition} \label{prop1}
Let $\tau_i$ be the inter-arrival time between the $i$-th and the $i-1$-th jumps of the Markov process $X$.
Under Assumptions \ref{ass6}, \ref{ass7}, \ref{ass8p2}, \ref{ass8add}, \ref{ass10}, \ref{ass11} and \ref{ass12}, almost surely, we have $$\frac{1}{n}\sum_{i=1}^n\tau_i  \underset{n \to \infty}{\to} \E_{\pi^{**}}[\tau(1)],$$ with $\pi^{**}$ the invariant distribution of $(q(n),c(n),\tau(n))$.
\end{proposition}

\begin{proof}
First we show that the Markov chain $(q(n),c(n),\tau(n))$ is positive Harris recurrent. For the Markov chain $(q(n),c(n))$, we have already proved that a coercive function $V$ can be found, such that the following drift condition is satisfied for some $a>0$ and $L<\infty$: 
$$\Delta V([q,c]) \leq -aV([q,c]) + L.$$ Now, for the Markov chain $(q(n),c(n),\tau(n))$, take $$V^*([q,c,t]) = V([q,c]) + t.$$ With obvious notation, we have 
\begin{eqnarray}
\Delta V^*([q,c,t])&=& \Delta V([q,c]) + \E_{[q,c,t]}[\tau] - t \nonumber\\
&\leq& -aV([q,c]) + L + \E_x[\tau] - t. \nonumber
\end{eqnarray}
Taking $a' = min(a,1)$ and $L' = L+1/m$, using Assumption \ref{ass12}, we get 
\begin{eqnarray}
\Delta V^*([q,c,t]) &\leq& -aV([q,c]) + L + 1/m - t \nonumber \\
&\leq& -a'[V([q,c])+t] + L + 1/m \nonumber \\
&\leq& -a'V^*([q,c,t]) + L'.\nonumber
\end{eqnarray}
Now let $C = \{([q,c,t]),~a'V^*([q,c,t]) \leq 2L'\}$. We have 
\begin{eqnarray}
\Delta V^*([q,c,t]) &\leq& -\frac{a'}{2}V^*([q,c,t]) + L'\one_{[q,c,t] \in C}.\nonumber
\end{eqnarray}
According to Lemma \ref{lem1}, the set $C$ is petite. Thus we can apply Theorem 15.0.1 in \citet*{Meyn2009} to deduce that the Markov chain $(q(n),r(n),\tau(n))$ is positive recurrent and thus admits an invariant measure. \\

Now remark that the function $V^*$ is unbounded off petite sets (using Lemma \ref{lem1} together with the fact that any subset of a petite set is itself petite). Consequently, Theorem 9.1.8 in \citet*{Meyn2009} enables us to obtain that the Markov chain $(q(n),c(n),\tau(n))$ is Harris recurrent. Therefore it is positive Harris recurrent.\\

We end the proof thanks to Theorem 17.0.1 from \citet*{Meyn2009}.
\end{proof}

\subsubsection{End of the proof of Theorem \ref{theo_scaling}}

We have $$\tilde{S}^{(n)}(t)= \frac{Z(\lfloor tn/\E_{\pi^{**}}[\tau(1)]\rfloor)}{\sqrt{n}} + \big(\frac{Z(N(nt))}{\sqrt{n}} - \frac{Z(\lfloor tn/\E_{\pi^{**}}[\tau(1)]\rfloor)}{\sqrt{n}}\big).$$ According to Proposition \ref{prop1}, the sequence of processes $N(nt)/n$ converges to $t/\E_{\pi^{**}}[\tau(1)]$. Moreover, the limit of $Z(\lfloor tn/\E_{\pi^{**}}[\tau(1)]\rfloor)/\sqrt{n}$ is continuous. Thus, using Skorohod representation theorem together with continuity properties in Skorohod topology, see Proposition VI.2.1 in \cite{jacod2003limit}, we get that the second term on the right hand side of the above equality tends to zero. Finally, from Theorem \ref{theo41}, we get $$ \frac{Z(\lfloor tn/\E_{\pi^{**}}[\tau(1)]\rfloor)}{\sqrt{n}} \overset{n \to \infty}{\to} \frac{\sigma}{\sqrt{\E_{\pi^{**}}[\tau(1)]}}B_t.$$ Combining these two results, we obtain the weak convergence of the rescaled price process to a Brownian motion with variance $\sigma^2/\E_{\pi^{**}}[\tau(1)]$, which concludes the proof.

\bibliography{ref}	

\begin{thebibliography}{17}
\newcommand{\enquote}[1]{``#1''}
\expandafter\ifx\csname natexlab\endcsname\relax\def\natexlab#1{#1}\fi

\bibitem[{Abergel and Jedidi(2011)}]{abergel2011mathematical}
Abergel, F. and Jedidi, A. (2011), \enquote{A mathematical approach to order
  book modelling,} in \textit{Econophysics of Order-driven Markets}, Springer,
  pp. 93--107.

\bibitem[{Abergel and Jedidi(2015)}]{abergel2015}
--- (2015), \enquote{Long time behaviour of a Hawkes process-based limit order
  book,} \textit{Preprint}.

\bibitem[{Bayer et~al.(2015)Bayer, Horst, and Qiu}]{Bayer2015}
Bayer, C., Horst, U., and Qiu, J. (2015), \enquote{A functional limit theorem
  for limit order books with state dependent price dynamics,} \textit{Preprint
  arXiv:1405.5230}.

\bibitem[{Billingsley(2009)}]{billingsley2009convergence}
Billingsley, P. (2009), \textit{Convergence of probability measures}, vol. 493,
  John Wiley \& Sons.

\bibitem[{Bouchaud et~al.(2002)Bouchaud, M{\'e}zard, and
  Potters}]{Bouchaud2002statistical}
Bouchaud, J.-P., M{\'e}zard, M., and Potters, M. (2002), \enquote{Statistical
  properties of stock order books: empirical results and models,}
  \textit{Quantitative finance}, 2, 251--256.

\bibitem[{Cont and De~Larrard(2013)}]{Cont2010}
Cont, R. and De~Larrard, A. (2013), \enquote{{Price dynamics in a Markovian
  limit order market},} \textit{SIAM Journal on Financial Mathematics}, 4,
  1--25.

\bibitem[{Cont et~al.(2010)Cont, Stoikov, and Talreja}]{cont2010stochastic}
Cont, R., Stoikov, S., and Talreja, R. (2010), \enquote{{A stochastic model for
  order book dynamics},} \textit{Operations research}, 58, 549--563.

\bibitem[{Foucault(1999)}]{foucault1999order}
Foucault, T. (1999), \enquote{Order flow composition and trading costs in a
  dynamic limit order market,} \textit{Journal of Financial markets}, 2,
  99--134.

\bibitem[{Huang et~al.(2015)Huang, Lehalle, and
  Rosenbaum}]{huang2013simulating}
Huang, W., Lehalle, C.-A., and Rosenbaum, M. (2015), \enquote{Simulating and
  analyzing order book data: The queue-reactive model,} \textit{Journal on the
  American Statistical Association}, 110, 107--122.

\bibitem[{Jacod and Shiryaev(2003)}]{jacod2003limit}
Jacod, J. and Shiryaev, A.~N. (2003), \textit{Limit theorems for stochastic
  processes}, Springer-Verlag, Berlin.

\bibitem[{Lachapelle et~al.(2014)Lachapelle, Lasry, Lehalle, and
  Lions}]{Lachapelle2014}
Lachapelle, A., Lasry, J.-M., Lehalle, C.-A., and Lions, P.-L. (2014),
  \enquote{Efficiency of the Price Formation Process in Presence of High
  Frequency Participants: a Mean Field Game analysis,} \textit{Preprint
  arXiv:1305.6323}.

\bibitem[{Lakner et~al.(2013)Lakner, Reed, and Stoikov}]{Lakner2013}
Lakner, P., Reed, J., and Stoikov, S. (2013), \enquote{High Frequency
  Asymptotics for the Limit Order Book,} \textit{Preprint}.

\bibitem[{Meyn and Tweedie(1993)}]{meyn1993stability}
Meyn, S.~P. and Tweedie, R.~L. (1993), \enquote{{Stability of Markovian
  processes III: Foster-Lyapunov criteria for continuous-time processes},}
  \textit{Advances in Probability}, 25, 518--548.

\bibitem[{Meyn and Tweedie(2009)}]{Meyn2009}
--- (2009), \textit{Markov chains and stochastic stability}, Cambridge
  University Press.

\bibitem[{Parlour(1998)}]{parlour1998price}
Parlour, C.~A. (1998), \enquote{Price dynamics in limit order markets,}
  \textit{Review of Financial Studies}, 11, 789--816.

\bibitem[{Ro{\c{s}}u(2009)}]{rocsu2009dynamic}
Ro{\c{s}}u, I. (2009), \enquote{A dynamic model of the limit order book,}
  \textit{Review of Financial Studies}, 22, 4601--4641.

\bibitem[{Smith et~al.(2003)Smith, Farmer, Gillemot, and
  Krishnamurthy}]{smith2003statistical}
Smith, E., Farmer, D.~J., Gillemot, L., and Krishnamurthy, S. (2003),
  \enquote{{Statistical theory of the continuous double auction},}
  \textit{Quantitative Finance}, 3, 481--514.

\end{thebibliography}

\end{document}